\documentclass[11pt]{article}
\usepackage{mathrsfs}
\usepackage{amssymb}
\usepackage{amsfonts}
\usepackage{color}
\usepackage{amsthm}
\usepackage{amsmath}
\usepackage{ascmac}
\usepackage[top=35truemm, bottom=30truemm, left=30truemm, right=30truemm]{geometry}
\usepackage{fancybox}
\usepackage{latexsym}
\newtheorem{dfn}{Definition}
\newtheorem{thm}{Theorem}[section]

\title{Dissipative dynamics of non-interacting fermion systems\\ and conductivity}
\author{Kazuki Yamaga\thanks{Department of Nuclear Engineering, Kyoto University, yamaga.kazuki.62a@st.kyoto-u.ac.jp}}
\date{}

\begin{document}
\maketitle
\begin{abstract}
In this paper, Non-Equilibrium Steady State induced by electric field and the conductivity of non-interacting fermion systems under the dissipative dynamics is discussed. The dissipation is taken into account within a framework of the quantum dynamical semigroup introduced in  \cite{davies1977irreversible}. We obtain a formula of the conductivity for the stationary state, which is applicable to arbitrary potentials. Our formula gives a justification of an adiabatic factor which is often introduced in practical calculation using Kubo formula. In addition the conductivity of crystals (i.e. periodic potentials) is also discussed.
\end{abstract}
\section{Introduction}
Given a macroscopic quantum system, its thermodynamical properties such as energy, heat capacity and magnetization are obtained from the microscopic information about the Hamiltonian of the system by considering the Gibbs state. This theory is summarized as equilibrium quantum statistical mechanics. Although there are no unified theory for general non-equilibrium systems, in linear response regime there is a formula for response functions to perturbations, Kubo formula \cite{kubo1957}. An important application of linear response theory is the electric conductivity. Being applied the electric field, the system in a stable equilibrium state is supposed to settle in another stable steady state with non-vanishing electric current (non-equilibrium steady state). In fact, the electric conductivity is discussed in \cite{kubo1957} based on the above observation. However, from the rigorous point of view, there is a subtle point in this discussion as explained below:

If the Kubo formula (equation (5.12) in \cite{kubo1957}) is naively applied to the simplest case, independently moving electrons, then the electric conductivity $\sigma$ is infinite,
\[ \sigma=\infty. \]
The reason is understood as follows. The formula of linear response theory is derived from the Hamiltonian dynamics of the closed system. Since the system is closed, the velocity of  the electron which is accelerated by the electric field is monotonically increasing and goes to infinite in the long time limit. Thus the electric conductivity becomes infinite finally. The issue comes from the absence of the non-equilibrium steady state (NESS) in this dynamics. In fact, a small adiabatic factor is often introduced in the practical calculation without any physical justification. 

One may claim that this difficulty is due to the idealistic nature of the free electron system. In a realistic circumstance, systems are not free from their environments, and therefore the dynamics must be dissipative. The aim of this paper is to see how the Kubo formula is modified by considering the dissipative dynamics. We consider the non-interacting fermion systems and quantum dynamical semigroup introduced in \cite{davies1977irreversible}.

This paper is organized as follows. First, in section 2 we explain the quantum dynamical semigroup of Davies \cite{davies1977irreversible} and derive the formula of conductivity in the lattice systems. The formula is almost the same as the Kubo formula with an adiabatic facror. But the parameter here represents the strengthen of the dissipation. Thus our formula can be regarded as a physical justification of the adiabatic factor. In the models that the potential is $0$ ($V=0$) the current is obtained concretely. This is the topic of section 3. Finally in section 4 the conductivity of crystals (that is the system with periodic potential) is discussed. It turns out that in the low temperature and small dissipation limit, the conductivity is given as the integral of the velocity over the Fermi surface. In appendix, we will treat the free continuous model and show that the Drude formula is obtained.

\section{Non-interacting fermion systems and Kubo formula}
We consider non-interacting many body systems of fermionic particles on $d$-dimensional lattice $\mathbb{Z}^d$. One particle is described by the Hilbert space $l^2(\mathbb{Z}^d)$ and the Hamiltonian
\[ (h\phi)(x)=-\sum_{|x-y|=1}\phi(y)+V(x)\phi(x),\hspace{5pt}\phi\in l^2(\mathbb{Z}^d). \]
$V$ is a real valued function on $\mathbb{Z}^d$ called potential, and we assume $V$ is bounded. Then $h$ is a bounded self adjoint operator. For a finite subset $\Lambda_N=[-N,N]^d\cap\mathbb{Z}^d$\ ($N\in\mathbb{N}$), denote the corresponding projection by $P_N$. That is, $P_N\colon l^2(\mathbb{Z}^d)\to l^2(\mathbb{Z}^d)$ is defined as
\[ (P_N\phi)(x)=
\begin{cases}
\phi(x) & x\in\Lambda_N \\
0 & \mathrm{otherwise}.
\end{cases}
\]
Define $h_N=P_NhP_N$. Since $h_N$ is a finite-rank self-adjoint operator it is decomposed as follows:
\[ h_N=\sum_{n=1}^{|\Lambda_N|}\epsilon_n^{(N)}|\phi_n^{(N)}\rangle\langle\phi_n^{(N)}| \]
(here we used the Dirac notation of physics). $\epsilon_n^{(N)}$ and $\phi_n^{(N)}$ are eigenvalues and the corresponding eigenvectors with norm $1$ and orthogonal each other. When the eigenvalue is degenerated, this decomposition is not unique. But the way of decomposition dose not matter in the following discussion. 

Now we consider many particle systems on $\mathbb{Z}^d$. The many body system is described by the CAR algebra $\mathcal{A}^{\mathrm{CAR}}(l^2(\mathbb{Z}^d))$ (we denote simply $\mathcal{A}$  below) generated by the creation / annihilation operators $a^*(f)\ /\  a(f)$ satisfying the following canonical anti-commutation relations:
\[ \{a^*(f),a(g)\}=\langle g,f\rangle1,\hspace{10pt}f,g\in l^2(\mathbb{Z}^d) \]
\[ \{a(f),a(g)\}=0, \]
where $\{A,B\}=AB+BA$ and $\langle\cdot,\cdot\rangle$ is the inner product of $l^2(\mathbb{Z}^d)$. Suppose the particles are moving independently by the Hamiltonian $h$, then the automorphism of the non-interacting dynamics on the many body system $\alpha_t\colon\mathcal{A}\to\mathcal{A}$ (Heisenberg picture) is given as
\[ \alpha_t(a^{\#}(f_1)\cdots a^{\#}(f_n))=a^{\#}(e^{ith}f_1)\cdots a^{\#}(e^{ith}f_n), \]
where $a^{\#}$ stands for $a$ or $a^*$.

A state is given as a normalized positive linear functional on the algebra $\mathcal{A}$. An important class of states in the fermion systems is quasi-free state. A state $\omega$ is called a quasi-free state if the following conditions are satisfied: For any $n,m\in\mathbb{N},f_1,\cdots,f_n,g_1,\cdots,g_m\in l^2(\mathbb{Z}^d)$,
\[ \omega(a^{*}(f_n)\cdots a^{*}(f_1)a(g_1)\cdots a(g_m))=\delta_{nm}\mathrm{det}((\omega(a^*(f_i)a(g_j)))_{ij}). \]
The property of quasi-free state is completely determined by its two-point functions. Since the two-point function is always expressed by a positive operator $0\le R\le1$ on $l^2(\mathbb{Z}^d)$ as
\[ \omega(a^*(f)a(g))=\langle g,Rf\rangle, \]
the quasi-free state is completely decided by an operator $R$. Conversely, given an operator $R$ on $l^2(\mathbb{Z}^d)$ that satisfies $0\le R\le1$, we have a quasi-free state by the above relations. For example, an equilibrium state $\omega_{\beta,\mu}$ (KMS state) at inverse temperature $\beta$ and chemical potential $\mu$ of a non-interacting fermion system with one-particle Hamiltonian $h$ is the quasi-free state corresponding to an operator $f_{\beta,\mu}(h)$, where $f_{\beta,\mu}$ is the Fermi-Dirac distribution $f_{\beta,\mu}(\epsilon)=\frac{1}{1+e^{\beta(\epsilon-\mu)}}$.

Next we consider dissipative dynamics on $\mathcal{A}$. We use the quantum dynamical semigroup discussed in \cite{davies1977irreversible}. 

First let us recall the definition of the quantum dynamical semigroup in general settings. A quantum dynamical semigroup is a family of state transformations with the semigroup property.

\begin{dfn}
Let $\mathcal{A}$ be a unital C*-algebra. A family of linear maps on $\mathcal{A}$, $\{T_t\}_{t\ge0}$, is called a quantum dynamical semigroup if it satisfies the following conditions:
\begin{description}
\item[(1)] $T_t$ is a unital CP $\mathrm{(}$completely positive$\mathrm{)}$ map for every $t\ge0$.
\item[(2)] $T_0=\mathrm{id}$  $\mathrm{(}\mathrm{id}$ is an identity map on $\mathcal{A}\mathrm{)}$.
\item[(3)] $T_{t+s}=T_t\circ T_s$\ $\mathrm{(}t,s\ge0\mathrm{)}$.
\item[(4)] $\displaystyle\lim_{t\downarrow0}\|T_t(A)-A\|=0,\hspace{5pt}A\in\mathcal{A}$.
\end{description}
\end{dfn}
In (1), unital means the relation $T_t(I)=I$ and complete positivity is defined as follows; a linear map $T\colon\mathcal{A}\to\mathcal{A}$ is completely positive if for any $N\in\mathbb{N}$, $T\otimes\mathrm{id}_N\colon\mathcal{A}\otimes M(N,\mathbb{C})\to\mathcal{A}\otimes M(N,\mathbb{C})$ is a positive map, where $M(N,\mathbb{C})$ is the algebra of $N\times N$ complex matrices. (1) is the condition that each $T_t$ is a state transformation. (2) implies that at time $0$ the system dose not change. (3) is the semigroup property. Here we contain the strong continuity (4) in the definition.

Quantum dynamical semigroup is one of tools describing the dynamics of open quantum systems. This is an approximation because in general the dynamics of open systems does not have the semigroup property (3). If the system is coupled to a large thermal reservoir with small relaxation time, this approximation is expected to work well. Indeed it is shown that quantum dynamical semigroup is obtained in the weak coupling limit in the physically natural settings \cite{davies1974markovian}. 

One of the most important features of quantum dynamical semigroup is that it is determined by its generator $L$ defined as
\[ L(A)=\lim_{t\downarrow0}\frac{T_t(A)-A}{t} . \]
This is defined only for $A\in\mathcal{A}$ such that the limit (the norm limit) of the right hand side exits. The set of such elements $\mathcal{D}(L)$ (the domain of $L$) forms a dense subspace of $\mathcal{A}$.

Now we turn to the fermion systems on the lattice $\mathbb{Z}^d$ and introduce a model. Here we want to deal with the system coupled to a heat reservoir. The dynamics must drive the states to an equilibrium state determined by the reservoir. As the system itself is an infinite system, it is not a trivial task to construct such a model. One of the possible models was given by Davies. Davies considered in \cite{davies1977irreversible} the dissipative  dynamics on many body fermion systems which is described by the language of one-particle Hilbert space. Let $\delta$ be the generator of the non-interacting dynamics $\alpha_t$. Define a *-automorphism $\theta\colon\mathcal{A}\to\mathcal{A}$ by $\theta(a^{\#}(f))=-a^{\#}(f)$. For each $N\in\mathbb{N}$, define a linear map $L_N\colon\mathcal{A}\to\mathcal{A}$ by $\mathcal{D}(L_N)=\mathcal{D}(\delta)$ and 
\begin{eqnarray*}
L_N(A)&&=\delta(A) \\
&&+\lambda\sum_{n=1}^{|\Lambda_N|}f_{\beta,\mu}(\epsilon_n^{(N)})\left(2a(\phi_n^{(N)})\theta(A)a^*(\phi_n^{(N)})-a(\phi_n^{(N)})a^*(\phi_n^{(N)})A-Aa(\phi_n^{(N)})a^*(\phi_n^{(N)})\right) \\
&&+\lambda\sum_{n=1}^{|\Lambda_N|}\left(1-f_{\beta,\mu}(\epsilon_n^{(N)})\right)\left(2a^*(\phi_n^{(N)})\theta(A)a(\phi_n^{(N)})-a^*(\phi_n^{(N)})a(\phi_n^{(N)})A-Aa^*(\phi_n^{(N)})a(\phi_n^{(N)})\right),
\end{eqnarray*}
for $A\in\mathcal{D}(L_N)$. $\lambda$ is a positive real number representing the strengthen of dissipation. The first term of the right hand side corresponds to the Hamiltonian dynamics, the second term represents the creation of particles in the eigenstate of energy by the distribution $f_{\beta,\mu}$ and the third term means the annihilation of particles by the distribution $(1-f_{\beta,\mu})$.

This map $L_N$ generates a quantum dynamical semigroup, we denote it $\{T_t^N\}_{t\ge0}$. The limit $\displaystyle\lim_{N\to\infty}T_t^N(A)$ exits for each $t\ge0,A\in\mathcal{A}$ and define a quantum dynamical semigroup $T_t(A)=\displaystyle\lim_{N\to\infty}T_t^N(A)$.

It is shown that every state $\psi$ converges to  $\omega_{\beta,\mu}$, the equilibrium state at inverse temperature $\beta$ and chemical potential $\mu$, in the weak* topology by the dynamics $T_t$:
\[ \lim_{t\to\infty}\psi\circ T_t(A)=\omega_{\beta,\mu}(A),\hspace{5pt}A\in\mathcal{A} .\]
Thus, $T_t$ is considered to describe the dynamics of the system coupled to a thermal reservoir. For the detail of this dynamics, see \cite{davies1977irreversible}.

Now suppose the system is initially in the stable equilibrium state $\omega_{\beta,\mu}$ under the dynamics $T_t$ and at time $t=0$ applied the uniform electric field $E$ in the direction $e_1=(1,0,\cdots,0)\in\mathbb{Z}^d$. Then one-particle Hamiltonian is changed to $h_E=h-EQ_1$, where $Q_1$ is the position operator; The domain is $\mathcal{D}(Q_1)=\displaystyle\left\{\phi\in l^2(\mathbb{Z}^d)\left| \sum_{x\in\mathbb{Z}^d}|x_1|^2|\phi(x)|^2\right.\right\}$ and for $\phi\in\mathcal{D}(Q_1)$, $Q_1$ acts as
\[ (Q_1\phi)(x)=x_1\phi(x) \]
where $x_1$ is the first element of $x\in\mathbb{Z}^d$. $Q_1$ is an unbounded self adjoint operator. Thus $h_E$ is also an unbounded self adjoint operator with the domain $\mathcal{D}(h_E)=\mathcal{D}(Q_1)$. Let $\delta_E$ be the generator of the non-interacting dynamics on $\mathcal{A}$ with the one-particle Hamiltonian $h_E$, and define a quantum dynamical semigroup $T_t^{E,N}$ by replacing $\delta$ in the definition of $L_N$ to $\delta_E$. These maps also define a quantum dynamical semigroup $T_t^E$ in the limit $N\to\infty$. This dynamics maps a quasi-free state to another quasi-free state. Let $\phi$ be a quasi-free state with the operator $R$, then $\phi\circ T_t^E$ is also a quasi-free state and the corresponding operator is
\[ e^{-2\lambda t}e^{-ith_E}Re^{ith_E}+2\lambda\int^t_0e^{-2\lambda s}e^{-ish_E}f_{\beta,\mu}(h)e^{ish_E}ds. \]
Thus the system initially in the equilibrium state $\omega_{\beta,\mu}$ finally converges to the stable non-equilibrium steady state by the dynamics $T_t^E$ in the weak* topology (in fact the limit does not depend on the initial state). We denote this state by $\omega_{\beta,\mu}^{\lambda,E}$ (apart from the case that $E=0$, this limit depends on $\lambda$). $\omega_{\beta,\mu}^{\lambda,E}$ is a quasi-free state corresponding to the operator
\[ 2\lambda\int^\infty_0e^{-2\lambda s}e^{-ish_E}f_{\beta,\mu}(h)e^{ish_E}ds. \]

Now let us consider the electric current. The increase of the number of particles at site $x\in\mathbb{Z}^d$ per unit time is expressed as
\[ \delta(a_x^*a_x)=\sum_{|x-y|=1}i(a^*_xa_y-a^*_ya_x) .\]
Here we write simply $a^{\#}(\eta_x)$ as $a^{\#}_x$, where $\{\eta_x\}_{x\in\mathbb{Z}^d}$ is the standard basis of $l^2(\mathbb{Z}^d)$: $\eta_x(y)=\delta_{xy}$. Each term in the summation represents the number of particles flowing from the nearest neighbor site of $x\in\mathbb{Z}^d$ per unit time. Define the observable $\hat{j}_{1,x}$ representing the current at $x$ in the direction $e_1$ as the average of the current flowing form $x-e_1$ to $x$ and the current from $x$ to $x+e_1$;
\[ \hat{j}_{1,x}=\frac{i}{2}(a^*_xa_{x+e_1}-a^*_{x+e_1}a_x+a^*_{x-e_1}a_x-a^*_xa_{x+e_1}) \]

The current in the non-equilibrium steady state $\omega^{\lambda,E}_{\beta,\mu}$ is 
\[ j^\lambda_{1,x}(E;\beta,\mu)\equiv\omega_{\beta,\mu}^{\lambda,E}(\hat{j}_{1,x}) .\]
The following theorem, one of our main results, shows that the response of the current $j_{1,x}^\lambda(E;\beta,\mu)$ to the electric field $E$ is simply expressed by the formula using the information of the one-particle Hamiltonian.
\begin{thm}
The current $j_{1,x}^\lambda(E;\beta,\mu)$ is differentiable at $E=0$ and its derivative $\sigma_{1,x}^\lambda(\beta,\mu)$ (electric conductivity) is expressed as
\[ \sigma_{1,x}^\lambda(\beta,\mu)=\mathrm{Re}\int_0^\infty e^{-2\lambda s}\langle\eta_x,e^{-ish}i[Q_1,f_{\beta,\mu}(h)]e^{ish}v_1\eta_x\rangle ds, \]
where $\mathrm{Re}$ is the real part and $v_1=i[h,Q_1]$ is the velocity ($v_1$ is independent of the potential $V$).
\end{thm}

\begin{proof}
First, let us calculate the two point functions $\omega_{\beta,\mu}^{\lambda,E}(a^*_xa_y)$.
\[ \omega_{\beta,\mu}^{\lambda,E}(a^*_xa_y)=2\lambda\int^\infty_0e^{-2\lambda s}\langle e^{ish_E}\eta_y,f_{\beta,\mu}(h)e^{ish_E}\eta_x\rangle ds. \]
Since $\eta_x,\eta_y\in\mathcal{D}(Q_1)=\mathcal{D}(h_E)$, $\langle e^{ish_E}\eta_y,f_{\beta,\mu}(h)e^{ish_E}\eta_x\rangle$ is differentiable by $s$. Integrating by parts we get
\begin{eqnarray*}
\omega_{\beta,\mu}^{\lambda,E}(a^*_xa_y)&=&\langle\eta_y,f_{\beta,\mu}(h)\eta_x\rangle \\
&&+\int^\infty_0e^{-2\lambda s}\left(\langle ih_Ee^{ish_E}\eta_y,f_{\beta,\mu}(h)e^{ish_E}\eta_x\rangle+\langle e^{ish_E}\eta_y,f_{\beta,\mu}(h)ih_Ee^{ish_E}\eta_x\rangle\right)ds \\
&=&\omega_{\beta,\mu}(a^*_xa_y)\\
&&+iE\int^\infty_0e^{-2\lambda s}\left(\langle Q_1e^{ish_E}\eta_y,f_{\beta,\mu}(h)e^{ish_E}\eta_x\rangle-\langle e^{ish_E}\eta_y,f_{\beta,\mu}(h)Q_1e^{ish_E}\eta_x\rangle \right)ds.
\end{eqnarray*}
We have $e^{ish}\eta_z\in\mathcal{D}(Q_1)$\ ($z\in\mathbb{Z}^d$) and for any $\phi\in l^2(\mathbb{Z}^d)$, $\displaystyle\lim_{E\to0}e^{ish_E}\phi=e^{ish}\phi$. In addition, since $|\langle\eta_x,f_{\beta,\mu}(h)\eta_y\rangle|$ decays exponentially with respect to $|x-y|$ \cite{aizenman1998localization}, for any $\phi\in\mathcal{D}(Q_1)$, $f_{\beta,\mu}(h)\phi\in\mathcal{D}(Q_1)$ and $[Q_1,f_{\beta,\mu}(h)]\equiv Q_1f_{\beta,\mu}(h)-f_{\beta,\mu}(h)Q_1$ can be extended to a bounded operator. Thus
\[ \lim_{E\to0} \langle\eta_y,e^{-ish_E}[Q_1,f_{\beta,\mu}(h)]e^{ish_E}\eta_x\rangle=\langle\eta_y,e^{-ish}[Q_1,f_{\beta,\mu}(h)]e^{ish}\eta_x \rangle .\]
In addition, one can exchange the limit and the integral and obtain
\[ \lim_{E\to0}\frac{\omega_{\beta,\mu}^{\lambda,E}(a^*_xa_y)-\omega_{\beta,\mu}(a^*_xa_y)}{E}=\int^\infty_0e^{-2\lambda s}\langle\eta_y,e^{-ish}i[Q_1,f_{\beta,\mu}(h)]e^{ish}\eta_x\rangle ds. \]
The relation 
\[ v_l\eta_z=i[h,Q_1]\eta_z=i\eta_{z+e_1}-i\eta_{z-e_1} \]
and the definition of $\hat{j}_{1,x}$ concludes the formula.
\end{proof}
Here we would like to mention the relation between the original Kubo formula and our formula. According to the paper \cite{kubo1957}, when the perturbation $fA$ ($|f|\ll1$) is added to the Hamiltonian $H$, the change of the first order $\Delta B$ of quantity $B$ in NESS is 
\[ \Delta B=\frac{1}{i}\int^\infty_0\mathrm{Tr}[A,\rho_{\beta,\mu}]B(t)dt=i\int_0^\infty\mathrm{Tr}\rho_{\beta,\mu}[A,B(t)]dt \]
where $B(t)=e^{itH}Be^{-itH}$ and $\rho_{\beta,\mu}$ is the Gibbs state. Or adding the adiabatic factor $e^{-\epsilon t}$ to converges the integral,
\[ \Delta B=\lim_{\epsilon\downarrow0}i\int^\infty_0e^{-\epsilon t}\mathrm{Tr}\rho_{\beta,\mu}[A,B(t)]dt. \]
But as discussed in the introduction, even though the adiabatic factor is added, if one take the limit $\epsilon\downarrow0$, $\Delta B$ may diverse (for example the electric current in the free model). And the physical meaning of the parameter $\epsilon$ is not clear. Our formula is changed to 
\[ i\int^\infty_0e^{-2\lambda t}\omega_{\beta,\mu}\left(\left[\sum_{n\in\mathbb{Z}}na^*_na_n,\hat{j}_x(t)\right]\right)dt,\]
where $\hat{j}_x(t)=e^{itH}\hat{j}_xe^{-itH}$ ($H$ is the total Hamiltonian without perturbation). This is same as the Kubo formula with the adiabatic factor (not taking the limit $\lambda\downarrow0$) in the case $\displaystyle A=\sum_{n\in\mathbb{Z}}na^*_na_n,\ B=\hat{j}_x$. The difference between the original Kubo formula and our approach is summarized as follows:\\
{\bf original}
\begin{itemize}
\item considering the Hamiltonian dynamics of closed system
\item NESS and the convergence to it are not discussed
\item the adiabatic factor $e^{-\epsilon t}$ is artificial and the physical meaning of the parameter $\epsilon$ is not clear
\end{itemize}
{\bf ours}
\begin{itemize}
\item considering the dissipative dynamics and discussing the convergence to NESS
\item the factor $e^{-2\lambda t}$ emerges naturally from the dissipative model and the physical meaning of the parameter $\lambda$ is clear, the strengthen of the dissipation or the inverse of relaxation time.
\end{itemize}

Although in this paper we do not discuss the magnetic field, the formula can be extended easily to the case that the magnetic field is  present and it produces the TKNN formula \cite{thouless1982quantized, bellissard1994noncommutative} in the limit $\beta\to\infty,\ \lambda\downarrow0$.

\section{Solvable model}
In the previous section we derived the formula of electric conductivity for the general form of potential $V$. In fact in the case where $V=0$, one can calculate the current explicitly. Here for simplicity we restrict the discussion to one-dimensional systems. That is, the one-particle Hilbert space is $l^2(\mathbb{Z})$. The one-particle Hamiltonian for potential $V=0$ is the multiplication operator on the momentum space $L^2(-\pi,\pi)$. Precisely the one-particle Hamiltonian is given by the multiplication operator on $L^2(-\pi,\pi)$
\[ (\hat{h}\phi)(k)=(-\cos k)\phi(k),\hspace{5pt}\phi\in L^2(-\pi,\pi) \]
and the Fourier transformation $\mathcal{F}\colon l^2(\mathbb{Z})\to L^2(-\pi,\pi)$ by $h=\mathcal{F}^{-1}\hat{h}\mathcal{F}$. As discussed in section 2 the system finally converges to a unique steady quasi-free state with the operator
\[ R_{\beta,\mu}^{\lambda,E}\equiv2\lambda\int^\infty_0e^{-2\lambda s}e^{-ish_E}f_{\beta,\mu}(h)e^{ish_E}ds. \]
Here we will obtain the explicit form of $R_{\beta,\mu}^{\lambda,E}$. In the following we consider in the momentum space and identify $\mathcal{F}^{-1}h\mathcal{F}$ as $h$. Note that $Q$ is a differential operator on the momentum spaces, thus
\[ (e^{ish}\phi)(k)=e^{i\epsilon(k)s}\phi(k) \]
\[ (e^{isQ}\phi)(k)=\phi(k+s) ,\]
where $\epsilon(k)=-\cos k$. Using the product formula \cite{chernoff1968note}
\[ e^{ish_E}\phi=\lim_{n\to\infty}\left(e^{i\frac{s}{n}h}e^{-i\frac{s}{n}EQ}\right)^n\phi, \]
we obtain 
\[ (e^{ish_E}\phi)(k)=\exp\left(is\int_0^1\epsilon(k-Es\xi)d\xi\right)\phi(k-Es), \]
and 
\begin{eqnarray*}
\langle\psi,R_{\beta,\mu}^{\lambda,E}\phi\rangle&=&2\lambda\int dk\int^\infty_0ds e^{-2\lambda s}\overline{\psi(k-sE)}f_{\beta,\mu}(\epsilon(k))\phi(k-sE) \\
&=&2\lambda\int^\infty_0ds\int dk\overline{\psi(k)}e^{-2\lambda s}f_{\beta,\mu}(\epsilon(k+sE))\phi(k) \\
&=&\int dk\overline{\psi(k)}\left(2\lambda\int^\infty_0e^{-2\lambda s}f_{\beta,\mu}(\epsilon(k+sE))ds\right)\phi(k)dk.
\end{eqnarray*}
Thus $R_{\beta,\mu}^{\lambda,E}$ is the multiplication operator of the function
\[ \left(R_{\beta,\mu}^{\lambda,E}\right)(k)=2\lambda\int^\infty_0e^{-2\lambda s}f_{\beta,\mu}(\epsilon(k+sE))ds .\]
Now let us calculate the current. Note that it is independent of the site (so we denote the current by $j^\lambda(E;\beta,\mu)$) and corresponds to the integration of the velocity $\frac{d}{dk}\epsilon(k)=-\sin k$:
\begin{eqnarray*}
j^{\lambda}(E; \beta,\mu)&=&\frac{1}{2\pi}\int^\pi_{-\pi}(-\sin k)\left(2\lambda\int^\infty_0e^{-2\lambda s}f_{\beta,\mu}(\epsilon(k+sE))ds\right)dk \\
&=&\frac{1}{2\pi}\int^\pi_{-\pi}\left(2\lambda\int^\infty_0(-\sin(k-sE))e^{-2\lambda s}ds\right)f_{\beta,\mu}(\epsilon(k))dk \\
&=&\frac{2\lambda E}{4\lambda^2+E^2}\int^\pi_{-\pi}\frac{-\cos k}{1+e^{-\beta(\cos k+\mu)}}\frac{dk}{2\pi} .
\end{eqnarray*}

The current becomes $0$ in the limit $E\to0$ and $\lambda\downarrow0$ respectively. On the  other hand, the conductivity is 
\[ \sigma_{l}^\lambda(\beta,\mu)=\frac{1}{2\lambda}\int^\pi_{-\pi}\frac{-\cos k}{1+e^{\beta(-\cos k-\mu)}}\frac{dk}{2\pi} \]
and goes to infinite as $\lambda\downarrow0$. 

\section{Electric conductivity in crystals}
In this section, we consider the electrons in crystals, that is the electrons under periodic potentials. Suppose the potential $V$ has period $p_l\in\mathbb{N}$ in the direction $e_l$ ($l=1,2,\cdots,d$);
\[ V(x+p_1e_1)=V(x+p_2e_2)=\cdots=V(x+p_de_d)=V(x),\hspace{5pt}x\in\mathbb{Z}^d .\]
Take $\Lambda=\{m_1e_1+\cdots m_de_d \mid 0\le m_1<p_1,\cdots0\le m_d< p_d\}$ and $\mathcal{B}=\mathbb{R}^d/(p_1\mathbb{Z}\times\cdots\times p_d\mathbb{Z})=\left(-\frac{\pi}{p_1},\frac{\pi}{p_1}\right]\times\cdots\left(-\frac{\pi}{p_d},\frac{\pi}{p_d}\right]$. $\mathcal{B}$ is called the Brillouin zone. For periodic potentials we can use the Bloch theory. The Hilbert space $l^2(\mathbb{Z}^d)$ is decomposed as the following direct integral
\[ l^2(\mathbb{Z}^d)=\int_\mathcal{B}^\oplus\mathcal{H}_kdk, \]
where $\mathcal{H}_k=l^2(\Lambda)\simeq\mathbb{C}^{|\Lambda|}$. And by the periodicity of $h$,  this is decomposed as
\[ h=\int_\mathcal{B}^\oplus h_kdk, \]
where $h_k$ is an operator on $l^2(\Lambda)$ defined as 
\[ (h_k\phi)(x)=-\sum_{|x-y|=1}\phi(y)+V(x)\phi(x),\hspace{5pt}x\in\Lambda \]
with boundary conditions $\phi(x+p_je_j)=e^{ikp_j}\phi(x)$. By this decomposition we also have
\[ f_{\beta,\mu}(h)=\int_\mathcal{B}^\oplus f_{\beta,\mu}(h_k)dk. \]
Since the commutator with the position operator $Q_1$ has the same periodicity, it is decomposed and it is given by the derivative:
\[ i[Q_1,f_{\beta,\mu}(h)]=\int_\mathcal{B}^\oplus \partial_{k_1}f_{\beta,\mu}(h_k)dk. \]

The velocity $v_1$ is also decomposed as 
\[ v_1=i[h,Q_1]=\int_\mathcal{B}^\oplus v_{1,k}dk,\]
where $v_{1,k}=-\partial_{k_1}h_k$.

Consider the mean of the conductivity
\[ \sigma_1^\lambda(\beta,\mu)=\frac{1}{|\Lambda|}\sum_{x\in\Lambda}\sigma^\lambda_{1,x}(\beta,\mu), \]
then this is expressed as 
\[ \sigma^\lambda_1(\beta,\mu)=\mathrm{Re}\int^\infty_0e^{-2\lambda s}\left(\int_\mathcal{B}\mathrm{Tr}e^{-ish_k}\partial_{k_1}f_{\beta,\mu}(h_k)e^{ish_k}v_{1,k}dk\right)ds. \]

In the following we consider the low temperature and small dissipation regime. Here we impose some assumptions.

Let $\epsilon_k^n$ and $\psi_k^n$ be eigenvalues and eigenvectors of $h_k$ ($n=1,2,\cdots,|\Lambda|$). Suppose
\begin{itemize}
\item $h_k$ is nondegenerate for all $k\in\mathcal{B}$
\item the eigenvalues $\epsilon_k^n$ and eigenvectors $\psi_k^n$ of $h_k$ are in $C^2$-class.
\end{itemize}

\begin{thm}
Under the above assumptions, we have
\[ \sigma_1^\lambda(\beta,\mu)=\frac{1}{2\lambda}\sum_{n=1}^{|\Lambda|}\int_\mathcal{B}f_{\beta,\mu}(\epsilon_k^n)\partial^2_{k_1}\epsilon_k^ndk+O(\lambda) .\]
Especially in the low temperature limit $\beta\to\infty$, electric conductivity $\sigma^\lambda_1(\infty,\mu)$ is expressed as
\begin{eqnarray*}
\sigma^\lambda_1(\infty,\mu)&=&\frac{1}{2\lambda}\sum_{n=1}^{|\Lambda|}\int_{\partial\mathcal{B}_\mu^n} \langle\psi_k^n ,v_{1,k}\psi_k^n\rangle n_1(k)dk+O(\lambda) \\
&=&\frac{1}{2\lambda}\sum_{n=1}^{|\Lambda|}\int_{\partial\mathcal{B}_\mu^n} \partial_{k_1}\epsilon_k^n n_1(k)dk+O(\lambda)
\end{eqnarray*}
where $\mathcal{B}_\mu^n=\{k\in\mathcal{B}\mid\epsilon_k^n\le\mu \}$ and $\partial\mathcal{B}_\mu^n$ is the boundary of $\mathcal{B}_\mu^n$ and $n_l(k)$ is the $l$-th element of the unit normal vector at $k\in\partial\mathcal{B}_\mu^n$. If $\mu$ is in the band gap, $\mathcal{B}_\mu^n$ has no boundary and the above integral is $0$. 
\end{thm}
This formula means that the main contribution to the conductivity in low temperature regime is given by integral of the velocity over the Fermi surface. And if $\mu$ is in the band gap, the conductivity is almost $0$ (insulator).
\begin{proof}
Since $\mathrm{Tr}e^{-ish_k}\partial_{k_1}f_{\beta,\mu}(h_k)e^{ish_k}\partial_{k_1}h_k$ is bounded for $k$, by the Fubini theorem we have
\[ \sigma^\lambda_1(\beta,\mu)=\mathrm{Re}\int_\mathcal{B}\left(\int^\infty_0e^{-2\lambda s}\mathrm{Tr}e^{-ish_k}\partial_{k_1}f_{\beta,\mu}(h_k)e^{ish_k}\partial_{k_1}h_kds\right)dk. \]
Put $P_k^n=|\psi_k^n\rangle\langle\psi_k^n|$. As $h_k$ is nondegenerate, we get
\begin{eqnarray*}
&&\int^\infty_0e^{-2\lambda s}\mathrm{Tr}e^{-ish_k}\partial_{k_1}f_{\beta,\mu}(h_k)e^{ish_k}\partial_{k_1}h_kds \\
&=&\sum_{n=1}^{|\Lambda|}\sum_{m=1}^{|\Lambda|}\int^\infty_0e^{-(2\lambda+i\epsilon_k^n-i\epsilon_k^m)s}ds\mathrm{Tr}P_k^n\partial_{k_1}f_{\beta,\mu}(h_k)P_k^m\partial_{k_1}h_k \\
&=&\frac{1}{2\lambda}\sum_{n=1}^{|\Lambda|}\mathrm{Tr}P_k^n\partial_{k_1}f_{\beta,\mu}(h_k)P_k^n\partial_{k_1}h_k \\
&&+\sum_{n\neq m}\frac{1}{2\lambda+i(\epsilon_k^n-\epsilon_k^m)}\mathrm{Tr}P_k^n\partial_{k_1}f_{\beta,\mu}(h_k)P_k^m\partial_{k_1}h_k.
\end{eqnarray*}
We estimate the right hand side using the following equations. In the sequel we skip the index $k$ and write $\partial_{k_1},\epsilon_k^n$ as $\partial,\epsilon^n$ for short. 

By differentiating $P^n=P^nP^n$, we have $\partial P^n=(\partial P^n)P^n+P^n\partial P^n$. Multiplying $P^n$ from both side, $P^n(\partial P^n)P^n=0$. Since $P^nP^m=0$ for $n\neq m$, we have $(\partial P^n)P^m+P^n\partial P^m=0$. Multiplying $P^n$ from both side, $P^n(\partial P^m)P^n=0$. From these equations we obtain
\[ \sum_{n=1}^{|\Lambda|}\mathrm{Tr}P^n\partial f_{\beta,\mu}(h)P^n\partial h=\sum_{n=1}^{|\Lambda|}\partial f_{\beta,\mu}(\epsilon^n)\partial\epsilon^n. \]
And if $n\neq m$ we have
\begin{eqnarray*}
P^n\partial f_{\beta,\mu}(h)P^m&=&\sum_{l=1}^{|\Lambda|}f_{\beta,\mu}(\epsilon^l)P^n(\partial P^l)P^m \\
&=&f_{\beta,\mu}(\epsilon^m)P^n(\partial P^m)P^m+f_{\beta,\mu}(\epsilon^n)P^n(\partial P^n)P^m \\
&=&f_{\beta,\mu}(\epsilon^m)P^n\partial P^m+f_{\beta,\mu}(\epsilon^n)(\partial P^n)P^m.
\end{eqnarray*}
In addition, we obtain $P^m\partial hP^n=\epsilon^m(\partial P^m)P^n+\epsilon^nP^m\partial P^n$ similarly. Therefore we have
\begin{eqnarray*}
\mathrm{Tr}P^n\partial f_{\beta,\mu}(h)P^m\partial h&=&f_{\beta,\mu}(\epsilon^m)\epsilon^m\mathrm{Tr}P^n(\partial P^m)^2+f_{\beta,\mu}(\epsilon^m)\epsilon^n\mathrm{Tr}P^n(\partial P^m)P^m\partial P^n \\
&&+f_{\beta,\mu}(\epsilon^n)\epsilon^m\mathrm{Tr}(\partial P^n)P^m(\partial P^m)P^n+f_{\beta,\mu}(\epsilon^n)\epsilon^n\mathrm{Tr}P^m(\partial P^n)^2\\
&=&f_{\beta,\mu}(\epsilon^m)(\epsilon^m-\epsilon^n)\mathrm{Tr}P^n(\partial P^m)^2+f_{\beta,\mu}(\epsilon^n)(\epsilon^n-\epsilon^m)\mathrm{Tr}P^m(\partial P^n)^2.
\end{eqnarray*}
Using the above equations 
\[ \sum_{n\neq m}\frac{1}{2\lambda+i(\epsilon_k^n-\epsilon_k^m)}\mathrm{Tr}P_k^n\partial_{k_1}f_{\beta,\mu}(h_k)P_k^m\partial_{k_1}h_k=\sum_{n\neq m}\frac{4\lambda(\epsilon_k^n-\epsilon_k^m)}{4\lambda^2+(\epsilon_k^m-\epsilon_k^n)^2}f_{\beta,\mu}(\epsilon_k^m)\mathrm{Tr}P_k^n(\partial_{k_1} P_k^m)^2. \]
Since $\epsilon_k^n$ is continuous for $k$ and $\epsilon_k^n\neq\epsilon_k^m$ $(n\neq m)$, there is a positive constant $C$ such that for all $k\in\mathcal{B}$ and $n\neq m$, $|\epsilon_k^n-\epsilon_k^m|>C$. Therefore we obtain the $\mu$-independent upper bound
\[ \left|\mathrm{Re}\int_\mathcal{B}\sum_{n\neq m}\frac{1}{2\lambda+i(\epsilon_k^n-\epsilon_k^m)}\mathrm{Tr}P_k^n\partial_{k_1}f_{\beta,\mu}(h_k)P_k^m\partial_{k_1}h_k \right|\le\frac{4\lambda}{C}\sum_{n\neq m}\int_\mathcal{B}\mathrm{Tr}P_k^n(\partial_{k_1} P_k^m)^2 dk. \]

\[ \sigma_1^\lambda(\beta,\mu)=-\frac{1}{2\lambda}\sum_{n=1}^{|\Lambda|}\int_\mathcal{B}\partial_{k_1}f_{\beta,\mu}(\epsilon_k^n)\partial_{k_1}\epsilon_k^ndk+O(\lambda)=\frac{1}{2\lambda}\sum_{n=1}^{|\Lambda|}\int_\mathcal{B}f_{\beta,\mu}(\epsilon_k^n)\partial^2_{k_1}\epsilon_k^ndk+O(\lambda). \]
Put $\mathcal{B}^n_\mu=\{k\in\mathcal{B}\mid \epsilon_k^n\le\mu\}$, then 
\begin{eqnarray*}
\sigma_1^\lambda(\infty,\mu)=\lim_{\beta\to\infty}\sigma_1^\lambda(\beta,\mu)&=&\frac{1}{2\lambda}\sum_{n=1}^{|\Lambda|}\int_{\mathcal{B}^n_\mu}\partial^2_{k_1}\epsilon_k^ndk+O(\lambda), \\
&=&\frac{1}{2\lambda}\sum_{n=1}^{|\Lambda|}\int_{\partial\mathcal{B}^n_\mu}\partial_{k_1}\epsilon_k^nn_1(k)dk+O(\lambda).
\end{eqnarray*}

\end{proof}

In one-dimensional case ($d=1$), the absence of degeneracy of eigenvalues $\epsilon_k^n$ implies the fact that all the gaps are open. Thus, $\mu$ is either in the gap or in the only one band ($n$-th band). Furthermore $\epsilon_k^n$ is an even function and monotonically increasing or decreasing on $0\le k\le \frac{\pi}{p}$. Denote $k_\mu\in[0,\frac{\pi}{p}]$ the solution of $\epsilon_k^n=\mu$. The conductivity is as follows in each case
\[ \sigma^\lambda(\infty,\mu)=\frac{1}{\lambda}\left.\partial_k\epsilon_k^n\right|_{k=k_\mu}+O(\lambda)\ \left(\epsilon_k^n\mathrm{\ is\ increasing\ on\ }0\le k\le \frac{\pi}{p}\right) \]

\[ \sigma^\lambda(\infty,\mu)=-\frac{1}{\lambda}\left.\partial_k\epsilon_k^n\right|_{k=k_\mu}+O(\lambda)\ \left(\epsilon_k^n\mathrm{\ is\ decreasing\ on\ }0\le k\le \frac{\pi}{p}\right). \]

\section{Discussion}
In this paper we considered the conductivity of non-interacting lattice fermion system coupled to a heat reservoir. We dealt with this open system by the quantum dynamical semigroup discussed by Davies. The main result is the formula in Theorem 2.1 and the justification of an adiabatic factor, $e^{-2\lambda t}$, of Kubo formula. In our approach the parameter $\lambda$ is naturally emerged by considering the dissipative dynamics and has the physical meaning, the strengthen of the dissipation or the inverse of relaxation time.  

In the case of periodic potentials, we showed that the conductivity is given by the integral of the velocity over the Fermi surface in the low temperature and small dissipation limit.

\appendix
\section{Appendix: continuous model and Drude formula}
In this paper we discussed lattice models. One can also consider the dissipative dynamics introduced in section 2 on continuous models. In this appendix we discuss the free continuous model and show that Drude formula is obtained from this model. 

For simplicity, here we consider a one-dimensional system as in section 3. In free continuous model, the one-particle Hilbert space is $L^2(\mathbb{R})$ and the one-particle Hamiltonian is the Fourier transformation of the multiplication operator: 

\[ (\hat{h}\phi)(k)=\frac{k^2}{2}\phi(k),\hspace{5pt}\phi\in L^2(\mathbb{R}) \]

The domain of $h$ is $\mathcal{D}(h)=\{\phi\in L^2(\mathbb{R})\mid\int_\mathbb{R}k^4|(\mathcal{F}\phi)(k)|^2dk<\infty\}$. In the continuous model, both $Q$ and $h$ are unbounded and we have to consider the domain of $h-EQ$ carefully. The operator $h-EQ$ defined on $\mathcal{D}(h-EQ)=\mathcal{D}(h)\cap\mathcal{D}(Q)$ as
\[ (h-EQ)\phi=h\phi-EQ\phi,\hspace{5pt}\phi\in\mathcal{D}(h-EQ), \]
is essentially self-adjoint, that is, the closure of it (we denote $h_E$) is a self-adjoint operator.

As in section 2, one can introduce the dissipative dynamics. The NESS is a quasi-free state generated by the following operator on $L^2(\mathbb{R})$:
\[ 2\lambda\int^\infty_0e^{-2\lambda s}e^{-ish_E}f_{\beta,\mu}(h)e^{ish_E}ds. \]
As the calculation in section 3 shows, this operator is the multiplication operator on momentum space of the function:
\[ \left(R_{\beta,\mu}^{\lambda,E}\right)(k)=2\lambda\int^\infty_0e^{-2\lambda s}f_{\beta,\mu}(\epsilon(k+sE))ds. \]

Here we calculate the current density as the integral of the momentum $k$ (this corresponds to our definition of current in lattice models ).

\begin{eqnarray*}
j^{\lambda}(E; \beta,\mu)&=&\int_\mathbb{R}k\left(2\lambda\int^\infty_0e^{-2\lambda s}f_{\beta,\mu}(\epsilon(k+sE))ds\right)dk \\
&=&\int_\mathbb{R}\left(2\lambda\int^\infty_0(k-sE)e^{-2\lambda s}ds\right)f_{\beta,\mu}(\epsilon(k))dk \\
&=&\frac{E}{2\lambda}\int_\mathbb{R}\frac{1}{1+e^{\beta(k^2-\mu)}}dk.
\end{eqnarray*}
$\int_\mathbb{R}\frac{1}{1+e^{\beta(k^2-\mu)}}dk$ is the density of particles per unit volume, and we denote it by $\rho$. Writing the mass $m$ and the charge $q$ explicitly, the last equation becomes 
\[ j^{\lambda}(E; \beta,\mu)=\frac{1}{2\lambda}\frac{\rho q^2}{m}E .\]
Noting that $(2\lambda)^{-1}$ means the relaxation time, this is the same as the formula known as the Drude formula.

\end{document}